\documentclass[a4paper,10pt]{article}
\usepackage{geometry}
\usepackage{titlesec}
\titleformat*{\section}{\large\bfseries}
\titleformat*{\subsection}{\large\bfseries}
\titleformat*{\subsubsection}{\large\bfseries}
\titleformat*{\paragraph}{\normalsize\bfseries}
\titleformat*{\subparagraph}{\normalsize\bfseries}

\usepackage{microtype}
\usepackage{csquotes} %
\usepackage{ifxetex}
\ifxetex%
        \usepackage{amsthm,amsmath,amssymb}
        \makeatletter
                \newcommand{\blx@nowarnpolyglossia}{}
        \makeatother

        \usepackage{polyglossia}
        \setdefaultlanguage{english}

\else
        \RequirePackage[utf8]{inputenc}

        \RequirePackage{amsfonts}
        \RequirePackage{amsthm,amsmath}
        \RequirePackage{amssymb}

        \usepackage[english]{babel}
\fi

\usepackage{tikz}
\usetikzlibrary{calc,arrows,decorations,decorations.pathreplacing}
\usepackage{pgfplots}
\pgfplotsset{compat=newest}
\usepackage{float}
\floatstyle{plain}
\restylefloat{figure}
        
\usepackage{subcaption}
\usepackage[backend=biber, style=numeric, sorting=anyt, autopunct=true,
maxnames=3, maxalphanames=10, natbib=true, arxiv=abs, doi=true, url=false,
eprint=true]{biblatex}
\usepackage{hyperref}
\hypersetup{%
	unicode=true,          %
	pdftitle={Connectedness of the space of rooted tree-based networks},
	pdfauthor={Péter L. Erdős, Andrew Francis, Tamás Róbert Mezei},
	pdfsubject={Phylogenetic combinatorics},   %
	pdfkeywords={nearest-neighbor interchange, rooted NNI, rooted phylogenetic
	network, directed acyclic graph},
	colorlinks=false,       %
}

\usepackage{enumitem}
\usepackage{cleveref}

\newtheorem{thm}{Theorem}[section]

\newtheorem{lem}[thm]{Lemma}

\theoremstyle{definition}
\newtheorem{defn}[thm]{Definition}

\newtheorem{rem}[thm]{Remark}

\newcommand{\rNNI}{{r\negthinspace{}N\negthinspace{}N\negthinspace{}I}}
\newcommand{\TBN}{{T\negthinspace{}B\negthinspace{}N}}
\newcommand{\RP}{{R\negthinspace{}P}}

\newcommand{\ot}{\gets}

\title{\Large\bfseries Rooted NNI moves on tree-based phylogenetic networks}
\author{Péter L. Erdős\textsuperscript{1,}\footnote{PLE and TRM were supported
	in part by the National Research, Development and Innovation Office ---
NKFIH grant K~116769, KH~126853, and K~132930} \\ \texttt{\normalsize
peter.erdos@renyi.hu}\and Andrew Francis\textsuperscript{2}\\
\texttt{\normalsize A.Francis@westernsydney.edu.au}\and Tamás Róbert
Mezei\textsuperscript{1,\fnsymbol{footnote},}\footnote{Corresponding author}\\
\texttt{\normalsize tamas.robert.mezei@renyi.hu}}

\date{%
	\small $^{1}$Alfréd Rényi Institute of Mathematics (a Hungarian Academy
	of Sciences Centre of Excellence), Reáltanoda~u.\ 13--15, 1053
	Budapest, Hungary \\ $^{2}$ Centre for Research in Mathematics and Data
	Science, Western Sydney University, Sydney, Australia\\
        \vspace{10pt}
        \normalsize \today
        \vspace{-12pt}
}

\begin{document}
\maketitle

\begin{abstract}
	We show that the space of rooted tree-based phylogenetic networks is
	connected under rooted nearest-neighbor interchange ($\rNNI$) moves.
\end{abstract}

\section{Introduction}

Phylogenetic networks are a generalisation of phylogenetic trees that have
become widely used as ways to represent evolutionary histories, because they are
able to either capture uncertainty in the inference, or represent non-tree-like
evolutionary processes~\cite{doolittle1999phylogenetic,koonin2015turbulent} (see
also the texts~\cite{huson2010phylogenetic,steel2016phylogeny}).  Such processes
include hybridization, in which two species combine to produce a third, and
horizontal gene transfer, in which genetic material from one species is
incorporated into that of another (common in bacteria).

Despite these non-tree-like evolutionary events, evolution can still appear
``tree-like'', in the sense that it may be representable as having a broad,
underlying tree, with additional arcs (directed edges) between the arcs of the
tree.  This sense motivated the definition of a ``tree-based
network''~\cite{francis2015phylogenetic}.

Tree-based networks have become an active area of research because they capture
biological intuition and have many mathematical
characterisations~\citep{francis2018new,jetten2016nonbinary,zhang2016tree} and
connections to other well-studied properties (for example they are precisely the
``tree-child'' networks for which every embedded tree is a base
tree~\cite{semple2016phylogenetic}).

For many applications, it is important to be able to randomly move around a set
of phylogenetic networks, for instance when searching for a network that
maximises a likelihood, or has the highest parsimony score. Mechanisms that
allow such movement are important, as without them such sampling is very
difficult.

The \emph{nearest neighbour interchange (NNI)} is a local operation on a graph
that is widely used for moving around the space of trees or networks.  It was
introduced for phylogenetic trees in 1971~\cite{robinson1971comparison},
generalised to unrooted phylogenetic networks in
2016~\cite{huber2016transforming}, and for rooted networks shortly
after~\cite{gambette_rearrangement_2017} (where the move is called rNNI).  The
spaces of such trees and networks are connected under the relevant rNNI moves,
and this allows random walks within those spaces to search for optimal trees or
networks.

In this paper we prove that the space of \emph{rooted tree-based networks} is
connected under rNNI moves. This is the rooted analogue of the result
of~\citet{fischer2019space} showing the connectedness of (unrooted) tree-based
networks under NNI moves.  We also show that the space is  connected under
the recently introduced \emph{tail}-%
moves~\cite{janssen2018exploring}.

The paper begins by introducing necessary concepts in Section~\ref{s:prelims}.
We then explore the effect of rNNI moves on a tree-based network with some
technical results in Section~\ref{s:impact}, and finally prove connectedness in
Section~\ref{s:connectedness}.

There are many opportunities to extend this work.  For instance, extending the
understanding of tree-based networks as a space, it would be interesting to
extend the notion of \emph{tree-based rank}, introduced for unrooted networks
in~\citep{fischer2019treebased}, to the rooted case.  This would be another
generalization of the proximity measures for rooted tree-based networks
discussed in~\cite{francis2018new}.
Finally, there are many other useful classes of network that might be connected
under such rearrangements, including well-studied classes such as the tree-child
network, and the recently introduced \emph{orchard}
networks~\cite{erdos2019class}.

\section{Preliminaries}\label{s:prelims}

A \emph{rooted phylogenetic network} $N$ on $X$ is a directed, simple acyclic
graph with the following types of vertices:
\begin{itemize}
	\setlength{\itemsep}{0pt}
	\item a single vertex of out-degree 1 or 2 and in-degree 0 called the
		\emph{root};
	\item vertices of in-degree 1 and out-degree 0 called \emph{leaves},
		which are labelled bijectively by the elements of $X$;
	\item vertices of in-degree 1 and out-degree 2, called \emph{tree
		vertices};
	\item vertices of in-degree 2 and out-degree 1, called
		\emph{reticulation vertices}.
\end{itemize}
Write $V=V(N)$ for the set of vertices of $N$, and $E=E(N)$ for the set
of arcs (directed edges) of $N$.  For an arc $e=(u,v)\in E$, write
$s(e):=u$ and $t(e):=v$ for the \emph{source} and \emph{target} of $e$,
respectively.  If $(u,v)\in E(N)$, we say $N$ has \emph{an arc on
$\{u,v\}$}.

Rooted phylogenetic networks with the above properties are commonly called
\emph{binary}.  Denote the set of rooted phylogenetic networks on $X$ by
$\RP(X)$.  Throughout this paper phylogenetic networks will be taken to be both
rooted and binary unless otherwise stated.

A rooted phylogenetic network without reticulation vertices is actually a rooted
tree, hence it is called a rooted phylogenetic $X$-tree.

An arc $e=(u,v)$ of $N\in\RP(X)$ may be \emph{subdivided} by removing $e$, and
adding a new vertex $w$ and new arcs $(u,w)$ and $(w,v)$.  A network with a
subdivided arc is no longer a phylogenetic network because it contains a vertex
of degree 2.  In the other direction, a vertex $w$ of degree 2 may be
\emph{suppressed} by deleting it and its two incident arcs $(u,w)$ and $(w,v)$,
and adding the arc $(u,v)$ to the network.

A rooted phylogenetic network that has a spanning tree $T$ whose leaves are
precisely the leaves of $N$, is a \emph{tree-based
network}~\cite{francis2015phylogenetic}.  Such a spanning tree for a tree-based
network $N$ is called a \emph{support tree} for $N$.  Note that a support tree
for $N$ is generally not a phylogenetic $X$-tree, because it will have vertices
of degree 2 (unless $N$ is itself a tree, in which case $T=N$).  By
``suppressing'' the vertices of degree 2 in $T$, one obtains a phylogenetic
$X$-tree $\widehat T$ that is called a \emph{base tree} for $N$, in the sense
that $N$ may be obtained from $\widehat T$ by ``subdividing'' arcs of $T$ and
adding additional arcs, as in the original definition
in~\cite{francis2015phylogenetic}.

The set of tree-based networks is denoted $\TBN(X)\subseteq\RP(X)$.

Nearest neighbour interchange (NNI) operations defined on phylogenetic trees
have been used to explore the space of trees for half a
century~\cite{robinson1971comparison}.  They have recently been generalised to
unrooted phylogenetic networks~\cite{huber2016transforming}, and to rooted
networks~\cite{gambette_rearrangement_2017}, as in Definition~\ref{d:rNNI}.

\begin{defn}\label{d:rNNI}
	Suppose $N\in\RP(X)$ has arcs on $\{a,b\},\{b,c\},\{c,d\}$, for
	distinct vertices $a,b,c,d\in V(N)$.  A \emph{rooted nearest neighbour
	interchange} ($\rNNI$) move on $\{a,b\},\{b,c\},\{c,d\}$, replaces
	those arcs with arcs on $\{a,c\},\{b,c\},\{b,d\}$, with the following
	conditions:
	\vspace{-4pt}
	\begin{enumerate}
		\setlength{\itemsep}{0pt}
		\item the in-degrees and out-degrees of $a$ and $d$ are
			unchanged;
		\item the in-degrees and out-degrees of $b$ and $c$ remain 1 or
			2;
		\item\label{d:rNNI.acyclic} the network remains an acyclic
			phylogenetic network.
	\end{enumerate}
\end{defn}

\medskip

Note that~\eqref{d:rNNI.acyclic} precludes the network $N$
from containing arcs on the arcs $\{a,c\}$ and $\{b,d\}$.

An $\rNNI$ move is a local operation on a subgraph of $N$ of four vertices and
three arcs.  If $P$ and $Q$ are subgraphs of $N$ such that $|V(P)|=|V(Q)|=4$
and $|E(P)|=|E(Q)|=3$, we say that an $\rNNI$ move switches $P$ to $Q$ if it
changes $N$ to $N'$ where $E(N')=(E(N)\setminus E(P))\cup E(Q)$.

For the proof of the main result we will need the notion of a ``burl-rooted
tree'', defined as follows.  This is a rooted version of the networks with
``$k$-burls'' used in~\cite{janssen2019rearrangement}.

\begin{defn}\label{d:burl.rooted.tree}
A \emph{burl-rooted tree} is a rooted phylogenetic network $N$ with reticulation
vertices $b_1,\dots,b_k$ and root $\rho$ with the following properties:
\vspace{-4pt}
\begin{itemize}
	\setlength{\itemsep}{0pt}
	\item there is a path from $\rho$ to a leaf $\ell_1$ that consists only
		of the vertices $(\rho,b_1,\dots,b_k,\ell_1)$;
	\item all paths from $\rho$ to other leaves begin
		$(\rho,a_1,\dots,a_k,u,\dots)$ for tree vertices
		$a_1,\dots,a_k$; and
	\item $N$ contains arcs $(a_i,b_i)$ for $i=1,\dots,k$.
\end{itemize}
\end{defn}

The structure of a burl-rooted tree is illustrated in Figure~\ref{fig:burl}.

\begin{figure}[H]
	\centering
	\begin{tikzpicture}\footnotesize
		\begin{scope}
			\node[draw, circle, inner sep=0pt, minimum
				size=12pt] (r) at (0,0) {$\rho$};
			\node[draw, circle, inner sep=0pt, minimum
				size=12pt] (l1) at ($ (r)-(3,3) $)
				{$\ell_1$};
			\node[draw, circle, inner sep=0pt, minimum
				size=12pt] (u) at ($ (r)-(-3,3) $)
				{$u$};
			\def\k{6}

			\draw[thick,->] (r)--(l1);
			\draw[thick,->] (r)--(u);

			\foreach \i/\j in {1/1,2/2,4/k}
			{%
			\pgfmathsetmacro\ratio{(\k-\i)/\k};
		\node[fill=white,draw, circle, inner sep=0pt, minimum size=12pt]
				(t\i) at ( $\i/\k*(l1)+\ratio*(r)$ ) {$b_{\j}$};
			\node[fill=white,draw, circle, inner sep=0pt, minimum size=12pt]
				(s\i) at ( $\i/\k*(u) +\ratio*(r)$ ) {$a_{\j}$};
			\draw[->] (s\i) edge[bend left=20] node[midway,below]
				{$\j$} (t\i);
			}

			\node at ($ 0.5*(r)+0.25*(l1)+0.25*(u)-(0,.4) $) {$\vdots$};

			\node[draw, circle, inner sep=0pt, minimum
				size=12pt] (l2) at ($ (u)-(3,3) $)
				{$\ell_2$};
			\node[draw, circle, inner sep=0pt, minimum
				size=12pt] (lr) at ($ (u)-(-3,3) $)
				{$\ell_r$};

			\node[draw, circle, inner sep=0pt, minimum size=12pt]
				(l3) at ($ 5/6*(l2)+1/6*(lr) $) {$\ell_3$};
			\node[draw, circle, inner sep=0pt, minimum size=12pt]
				(l4) at ($ 4/6*(l2)+2/6*(lr) $) {$\ell_4$};
			\node at ($ 0.5*(l4)+0.5*(lr) $) {$\dots$};

			\node[draw, circle, inner sep=0pt, minimum size=12pt]
				(v) at ($ 5/6*(l2)+1/6*(u) $) {$v$};

			\draw[thick,->] (v) -- (l2);
			\draw[thick,->] (v) -- (l3);

			\draw[thick,->] ($ (lr)+2*(v)-2*(l3) $) -- (lr);
			\draw[thick,->] ($ (l4)+2*(v)-2*(l2) $) -- (l4);
			\draw[thick,->] ($ (v)+1.5*(v)-1.5*(l2) $) -- (v);

			\draw[thick,->] (u) -- ($ 4/6*(u)+2/6*(l2) $);
			\draw[thick,->] (u) -- ($ 4/6*(u)+2/6*(lr) $);

			\node at ($ 0.5*(u)+0.25*(l2)+0.25*(lr) $) {$\vdots$};

			\draw[dashed]	($ 1.2*(l2)-0.1*(u)-0.1*(lr) $) --
					($ 1.2*(u)-0.1*(lr)-0.1*(l2) $) --
					($ 1.2*(lr)-0.1*(l2)-0.1*(u) $) --
					($ 1.2*(l2)-0.1*(u)-0.1*(lr) $);
		\end{scope}
	\end{tikzpicture}
	\caption{The structure of a burl-rooted tier-$k$ tree. The $k$
	reticulation arcs join the vertices between $\rho-\ell_1$ and $\rho-u$.
	The vertices contained inside the dashed triangle induce a rooted binary
	tree with $r-1$ leaves.}\label{fig:burl}
\end{figure}
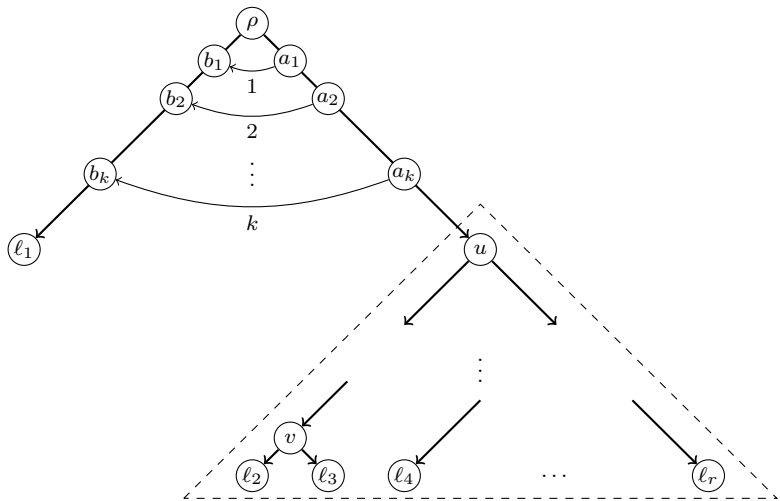

Finally we recall the definition of head and tail moves, introduced
in~\cite{janssen2018exploring}.

\begin{defn}\label{d:head.tail.move}
Let $e=(u,v)$ and $f$ be arcs in a rooted phylogenetic network $N$.  A
\emph{tail move} of $e$ to $f$ involves: deleting $e$; subdividing $f$ with a
new node $u'$; suppressing $u$; and adding the arc $(u',v)$.  A \emph{head move}
of $e$ to $f$ involves: deleting $e$; subdividing $f$ with a new node $v'$;
suppressing $v$; and adding the arc $(u,v')$.
\end{defn}

\section{The impact of \texorpdfstring{$\rNNI$}{rNNI} moves on
tree-based-ness}\label{s:impact}

\begin{lem}\label{lem:NNI} %
	Let $N$ be a rooted tree-based phylogenetic network with support tree $T$.
	Suppose $P:u\to v\to w\to z$ is a path of length 3 in $T$. Let $e,f\in
	E(N)\setminus E(P)$ be arcs incident to $v$ and $w$, respectively.
	If either
	\begin{enumerate}[label={$($\alph*$)$}]
		\setlength{\itemsep}{0pt}
		\item\label{lem:headNNI} $f$ is oriented away from $w$ and
			$e\neq vz$, or
		\item\label{lem:tailNNI} $e$ is oriented towards $v$ and
			$f\neq uw$, or
		\item\label{lem:treeNNI} $f\neq uw$ and $e\neq vz$, and $N$ does not contain a directed
			$t(e)\to s(f)$ path,
	\end{enumerate}
	then the $\rNNI$ move switching the path $P$ to $Q:u\to w\to v\to z$ is
	valid and the resulting network $N'$ is still tree-based.
	\begin{figure}[ht]
		\centering
		\begin{tikzpicture}
			\begin{scope}
				\node[draw, circle, inner sep=0pt, minimum
					size=16pt] (u) at (0,0) {$u$};
				\node[draw, circle, inner sep=0pt, minimum
					size=16pt] (v) at ($ (u)-(-1,1) $)
					{$v$};
				\node[draw, circle, inner sep=0pt, minimum
					size=16pt] (w) at ($ (u)-(-2,2) $)
					{$w$};
				\node[draw, circle, inner sep=0pt, minimum
					size=16pt] (z) at ($ (u)-(-3,3) $)
					{$z$};
				\draw[thick,->] (u)--(v);
				\draw[thick,->] (v)--(w);
				\draw[thick,->] (w)--(z);

				\draw (v) -- ($ (v)-(1.2,1.2) $) node[near
					end,above] {$e$};
				\draw (w) -- ($ (w)-(1.2,1.2) $) node[near
					end,above] {$f$};
				\node at (2,0) {$N$};
			\end{scope}
			\begin{scope}[xshift=7cm]
				\node[draw, circle, inner sep=0pt, minimum
					size=16pt] (u) at (0,0) {$u$};
				\node[draw, circle, inner sep=0pt, minimum
					size=16pt] (w) at ($ (u)-(-1,1) $)
					{$w$};
				\node[draw, circle, inner sep=0pt, minimum
					size=16pt] (v) at ($ (u)-(-2,2) $)
					{$v$};
				\node[draw, circle, inner sep=0pt, minimum
					size=16pt] (z) at ($ (u)-(-3,3) $)
					{$z$};
				\draw[thick,->] (u)--(w);
				\draw[thick,->] (w)--(v);
				\draw[thick,->] (v)--(z);

				\draw (v) -- ($ (v)-(1.2,1.2) $) node[near
					end,above] {$e$};
				\draw (w) -- ($ (w)-(1.2,1.2) $) node[near
					end,above] {$f$};
				\node at (2,0) {$N'$};
			\end{scope}
			\draw[->,very thick,decorate,decoration={snake,segment
				length=8pt,amplitude=2pt}] (4,-1.5)--(6,-1.5);
		\end{tikzpicture}
		\caption{An $\rNNI$ move which is valid if one of the three
		conditions of \Cref{lem:NNI} hold.}\label{fig:NNI}
	\end{figure}
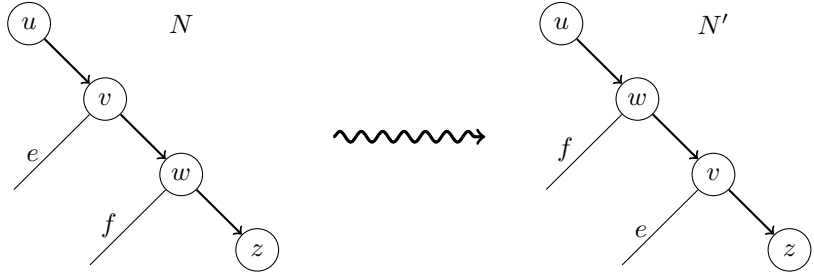
\end{lem}
\begin{rem}
The $\rNNI$ move $P\to Q$ simply relocates $w$ onto the $uv$ arc.
Depending on the orientation of $f$, this $\rNNI$ move is equivalent (up to isomorphism) to
a distance-1 head-move or tail-move.
\end{rem}
\begin{proof}
	Both $zv,wu\notin E(N)$, because either arc would make $N$ cyclic.
	If $f$ is oriented away from $w$, then $uw\notin E(N)$, because $w$ has
	total degree 3.  If $e$ is oriented towards $v$, then $vz\notin E(N)$,
	because $v$ has total degree 3.

	Thus the network $N'$ created by the $\rNNI$ move transforming $P$ into
	$Q$ is well-defined, but it might contain an oriented cycle.  Suppose
	there is an oriented cycle in $N'$, let the shortest one be $\vec C$.
	Note that this forces $z\neq \rho$, because the root $\rho$ has in-degree 0.

	Suppose first, that $f,wv,e\in\vec C$: then $f$ is oriented towards
	$w$, $e$ is oriented away from $v$, and there is a directed $t(e)\to
	s(f)$ path in $N'$, which is also present in $N$. In any case, we have
	a contradiction.  If both $e,f\in \vec C$, but $wv\notin \vec C$, then
	$\vec C$ is not the shortest oriented cycle, because we could shortcut
	through $wv$.  Because $e$ and $f$ cannot be both traversed by the
	cycle, $\vec{C}$ can be trivially shortened or extended by one arc to
	form an oriented cycle in $N$.

	Lastly, observe that $T'=T-E(P)+E(Q)$ is a support tree of $N'$.
\end{proof}

\begin{lem}\label{lem:rootNNI}
	Let $N$ be a rooted phylogenetic network with support tree $T$.
	Suppose $P:u\ot z\to v\to w$ is a subgraph of 3 arcs in $T$. If there
	is no $u\to v$ path in $N$ and $vu\notin E(N)$, then the $\rNNI$ move
	switching $P$ to $Q:u\ot v\ot z\to w$ is valid and the resulting
	network is still tree-based. The statement holds even if $z=\rho$.
	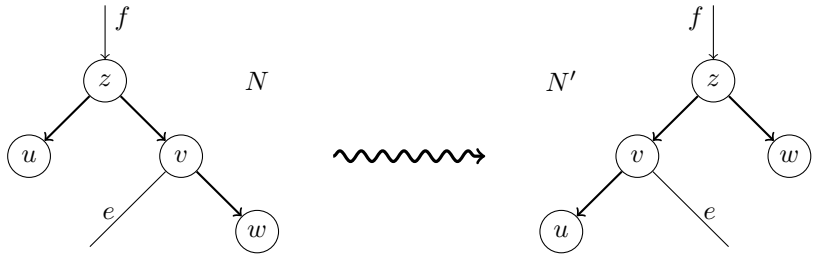
\begin{figure}[H]
		\centering
		\begin{tikzpicture}
			\begin{scope}
				\node[draw, circle, inner sep=0pt, minimum
					size=16pt] (r) at (0,0) {$z$};
				\node[draw, circle, inner sep=0pt, minimum
					size=16pt] (u) at ($ (r)-(1,1) $)
					{$u$};
				\node[draw, circle, inner sep=0pt, minimum
					size=16pt] (v) at ($ (r)-(-1,1) $)
					{$v$};
				\node[draw, circle, inner sep=0pt, minimum
					size=16pt] (w) at ($ (r)-(-2,2) $)
					{$w$};
				\draw[thick,->] (r)--(u);
				\draw[thick,->] (r)--(v);
				\draw[thick,->] (v)--(w);

				\draw (v) -- ($ (v)-(+1.2,1.2) $) node[near
					end,above] {$e$};
				\draw[->] ($ (r)+(0,1) $) -- (r)  node[near
					start, right] {$f$};
				\node at (2,0) {$N$};
			\end{scope}
			\begin{scope}[xshift=8cm]
				\node[draw, circle, inner sep=0pt, minimum
					size=16pt] (r) at (0,0) {$z$};
				\node[draw, circle, inner sep=0pt, minimum
					size=16pt] (u) at ($ (r)-(2,2) $)
					{$u$};
				\node[draw, circle, inner sep=0pt, minimum
					size=16pt] (v) at ($ (r)-(1,1) $)
					{$v$};
				\node[draw, circle, inner sep=0pt, minimum
					size=16pt] (w) at ($ (r)-(-1,1) $)
					{$w$};
				\draw[thick,->] (v)--(u);
				\draw[thick,->] (r)--(v);
				\draw[thick,->] (r)--(w);

				\draw (v) -- ($ (v)-(-1.2,1.2) $) node[near
					end,above] {$e$};
				\draw[->] ($ (r)+(0,1) $) -- (r)  node[near
					start, left] {$f$};
				\node at (-2,0) {$N'$};
			\end{scope}
			\draw[->,very thick,decorate,decoration={snake,segment
				length=8pt,amplitude=2pt}] (3,-1) -- (5,-1);
		\end{tikzpicture}
	\caption{A child $v$ of $z$ can move across $z$ into its other branch
	if there is no arc joining $\{u,v\}$. If $z=\rho$ then $f$ should be
	omitted from the picture.}\label{fig:rootNNI}
	\end{figure}
\end{lem}
\begin{proof}
	By the assumptions, there is no arc of any orientation between $u$ and
	$v$. Because $N$ is acyclic, $wz\notin E(N)$.
	Therefore the $\rNNI$ move
	switching $P$ to $Q$ produces a valid network.

	Suppose there is an oriented cycle in $N'$; let the shortest
	one be $\vec C$.

	Suppose first, that $e\in\vec C$: either $e$ is oriented towards $v$
	and $vu\in \vec C$, or $e$ is oriented away from $v$ and $zv\in\vec C$.
	In any case, this means that there is a $u\to v$ path in $N$.  If $f\in
	\vec C$ and $e\notin\vec C$, then $\vec{C}$ can be trivially shortened
	or extended by one arc to form an oriented cycle in $N$. If $e,f\notin
	\vec C$, then $\vec C$ is already an oriented cycle in $N$.

	Lastly, observe that $T'=T-E(P)+E(Q)$ is a support tree of $N'$.
\end{proof}

\section{The connectedness of the space of tree-based
networks}\label{s:connectedness}

\begin{lem}\label{lem:burl.rooted.conn}
The set of burl-rooted trees in tier $k$ is connected under $\rNNI$ moves.
\end{lem}
\begin{proof}
By definition, the burls of burl-rooted trees in tier $k$ are identical, and
they only differ by the trees below the burl (vertex $u$ in
Figure~\ref{fig:burl}).  Since the space of trees is connected under $\rNNI$
moves~\cite{robinson1971comparison}, one can be transformed into the other,
treating the vertex in position $u$ as the root.
\end{proof}

We can now prove our main theorem.
\begin{thm}\label{thm:TBNrNNIconn}
	$\TBN(X)$ is connected under $\rNNI$ moves.
\end{thm}
\begin{proof}
	We show that any tree-based network can be transformed into a
	{burl-rooted tree} (Definition~\ref{d:burl.rooted.tree}), and use the
	fact that it is possible to
	move between any two networks in that form
	(Lemma~\ref{lem:burl.rooted.conn}).

	We fix an arbitrary tree-based network $N$, and a support tree $T$ for
	$N$. In each step we need to cover four cases regarding $N$ and $T$ and
	their root $\rho$:
	\begin{enumerate}[label={(\Alph*)}]
		\setlength{\itemsep}{0pt}
		\item $\rho$ has out-degree 1 in $T$;
		\item $\rho$ has out-degree 2 in $T$ and on both sides of the
			root there are branching points in $T$;
		\item $\rho$ has out-degree 2 in $T$, but on one side of the
			root there are no branching points in
			$T$;
		\item $\rho$ has out-degree 2 in $T$ and $T$ is path (in this case $|X|=2$).
	\end{enumerate}
	Note, that multiple support trees may exist for a fixed tree-based
	network $N$. Furthermore, the degree of $\rho$ might be $1$ in one
	support tree, and $2$ in another, which means that in the first case
	$\rho$ must be the source of a reticulation arc.

	\subsection{Case (A), when \texorpdfstring{\boldmath{$|X|=1$}}{|X|=1}.}

	Although this is a degenerate case, we need to deal with it for the sake
	of completeness. There is no burl-rooted tree when there is only one
	leaf. Let $e=\rho v$ be the reticulation arc incident on the root.
	\Cref{lem:NNI}\ref{lem:headNNI} applies to the source of reticulation
	that is below $v$ and closest to it. Therefore, we can move every source
	of reticulation between $\rho$ and $v$ one-by-one. Next, via
	\Cref{lem:NNI}\ref{lem:tailNNI}, we can move every target of
	reticulation below $v$ similarly, in an appropriate order. Lastly, we
	can freely permute the sources between $\rho$ and $v$ via
	\Cref{lem:NNI}\ref{lem:headNNI}, and similarly, we can permute the
	targets below $v$ freely via \Cref{lem:NNI}\ref{lem:tailNNI}. Thus it is
	clear that any two networks of tier-$k$ that have exactly one leaf each
	are connected via $\rNNI$ moves.

	\subsection{Case (A), when \texorpdfstring{\boldmath{$|X|\ge
	2$}}{|X|>=2}.}\label{sec:caseA2}

	Let $b_1$ be the branching point in $T$ which is the closest to $\rho$
	(in both $T$ and $N$). Let $\ell$ be an arbitrarily chosen leaf, and
	let $b_2$ be the closest branching point to it. (We may have
	$b_1=b_2$). Let $d_{T}(x,y)$ be the undirected distance between $x$ and $y$ in the
	support tree $T$ for $N$. Define the quantity
	\[\varTheta_{N,T}:=\sum_{f\in E(N)\setminus E(T)}
	d_{T}(\rho,s(f))+\sum_{e\in E(N)\setminus E(T)}
	d_{T}(t(e),\ell). \]
	Let the number of reticulation arcs be $\mathrm{\tau}$. We claim that
	via $\rNNI$ moves we can reduce $\varTheta$ to
	$\binom{\tau}{2}+\binom{\tau+1}{2}=\tau^2$ (note, that there
	is a reticulation arc whose source is $\rho$), ie.,\ in the desired
	network, a vertex $v$ is
	\begin{itemize}
		\setlength{\itemsep}{0pt}
		\item the source of a reticulation if and only if $v=\rho$
			or $v$ is between $\rho$ and $b_1$ on the
			support tree,
		\item the target of a reticulation if and only if $v$ is between
			$b_2$ and $\ell$ on the support tree.
	\end{itemize}

	Suppose $f$ is a reticulation arc for which
	$d_{T}(\rho,s(f))>d_{T}(\rho,b)$, and the left hand side is
	minimal wrt.\ $f$. \Cref{lem:NNI}\ref{lem:headNNI} applies to our case,
	because $e=vz$ contradicts the minimality assumptions on $f$. For the
	same reason, the $\rNNI$ move specified by \Cref{lem:NNI} decreases
	the first sum in $\varTheta$ by $1$.

	If such an $f$ does not exist, but $\varTheta>\tau^2$, then
	there exists an $e$ such that $t(e)$ is not between $b_2$ and $\ell$.
	Choose the $e$ for which $d_{T}(t(e),\ell)>d_{T}(b_2,\ell)$,
	and the left hand side is minimal wrt.\ $e$. We have three cases.

	\begin{enumerate}[label={$($\roman*$)$}]
		\setlength{\itemsep}{0pt}
		\item If the undirected $t(e)\to \ell$ path in $T$
			starts on an out-arc of $t(e)$, then
			\Cref{lem:NNI}\ref{lem:tailNNI} applies to $e$ (with
			the same name), because $f=uw$ contradicts the
			minimality assumption on $e$. The $\rNNI$ move
			specified by \Cref{lem:NNI} decreases $\varTheta$ by
			$1$.
		\item If the undirected $t(e)\to \ell$ path in $T$
			starts on an in-arc of $t(e)$, and the next arc is in
			opposite orientation (see the bold arcs in $N$ in
			\Cref{fig:rootNNI}), then we apply \Cref{lem:rootNNI}
			to $e$. The conditions of the lemma are satisfied,
			because the sources of reticulation arcs are closer to
			the root than the parent of $t(e)$ in $T$. By the
			minimality assumption on $e$, the $\rNNI$ move reduces
			$\varTheta$ by at least $1$.
		\item If the undirected $t(e)\to \ell$ path in $T$
			starts on an in-arc of $t(e)$, and the next arc is in
			the same orientation (see graph $N'$ in
			\Cref{fig:NNI}), then we apply the $\rNNI$ move of
			\Cref{lem:NNI}\ref{lem:treeNNI} to $e$, but with the
			labels of arcs $e$ and $f$ exchanged. Because $b_1\neq
			v,w,z$ in the setup of \Cref{lem:NNI}, the conditions
			of the lemma are satisfied. The $\rNNI$ move decreases
			$\varTheta$ by $1$ (by the minimality assumption on
			$e$). %
	\end{enumerate}

	We may assume now that $\varTheta_{N,T}=\tau^2$. Let $e$ be the
	reticulation arc whose source is the root $\rho$. Via a couple of
	$\rNNI$ moves provided by \Cref{lem:NNI}\ref{lem:treeNNI}, we may assume
	that $t(e)$ is the child of $b_2$ in $T$ (while keeping
	$\varTheta=\tau^2$). Change the tree base while keeping the network $N$
	intact: let $T'=T-b_2t(e)+e$. This a support tree for $N$, because $b_2$
	is a branching vertex in $T$.

	Although $e$ is no longer a reticulation arc, $b_2t(e)$ becomes one.
	If $b_1=b_2$, we have a burl-rooted tree. Otherwise, $b_2$ can be moved
	to the path between $\rho$ and $b_1$ via
	\Cref{lem:NNI}\ref{lem:headNNI}.  Lastly, note that the choice of leaf
	$\ell$ on the burl has been arbitrary.

	\subsection{Cases (B) and (C).}\label{sec:caseBC}

	Suppose $b_1$ is a branching vertex closest to $\rho$ in $T$. Let $\ell$
	be an arbitrary leaf below $b_1$ in $T$, and let $b_2$ be the branching
	point it is joined to in $T$.  On the branch of the root containing
	$b_1$ we may perform the procedure outlined in the previous Case (A)
	until $\varTheta$ is reduced to its minimum (counting only sources or
	targets of reticulations on the branch of $b_1$). We have to check that
	reticulation arcs that join the two main branches (originating at the
	root) do not interfere with the previously described $\rNNI$ sequence.
	This is easily seen to be the case.

	\medskip

	We will reduce this case to Case (A) now. Let $u,v$ be the children of
	$\rho$ such that $v$ is on the same branch as $b_1$ in $T$.  Let $t(e)$
	be the target of reticulation which is the child of $b_2$ on $T$.

	If $vu\in E(N)$, then we change the support tree of $N$ to
	$T'=T-\rho u+vu$, and the reduction to Case (A) is done.

	If $vu\notin E(N)$, then the root can be moved down along the $\rho\to
	t(e)$ path in $T$ until $\rho$ is between $b_2$ and $t(e)$; all we
	need to do is check that the conditions of \Cref{lem:rootNNI} are
	satisfied at each step. Because a directed path cannot traverse the
	root and all of the targets of reticulation arcs are below $b_2$ in
	the branch of $v$, the conditions are satisfied. Once we have $t(e)\ot
	\rho\to b_2$, we change the support tree to
	$T'=T-\rho t(e)+e$. We have completely reduced this case to
	Case (A).

	\subsection{Case (D)}

	Let $u,v$ be the two children of $\rho$. Without loss of generality, we
	may assume that $v$ or one of the vertices below it in the support tree
	is a target of reticulation. If $vu\in E(N)$, we can rewire the support
	tree through $vu$ and reduce this case to Case (A). If there is a $u\to
	v$ path in $N$, then $v$ must be the target of a reticulation arc $e$,
	such that $s(e)$ is below $u$ in $T$ (otherwise $N$ would contain an
	oriented cycle). Via \Cref{lem:NNI}\ref{lem:headNNI}, we may assume that
	$e=uv$, and we can rewire the support tree through $uv$ to reduce this
	case to Case (A).

	If $vu\notin E(N)$ and there is no $u\to v$ path in $N$, we can perform
	the $\rNNI$ move described by \Cref{lem:rootNNI}. By repeating the
	argument, we may assume that $v$ is the target of a reticulation arc,
	in which case we are done (as above).
\end{proof}

\subsection{Connectedness under distance-1 tail-moves}

\begin{thm}\label{thm:TBNtailmoveconn}
	$\TBN(X)$ is connected under distance-1 tail-moves.
\end{thm}
\begin{proof}
In the proof of \Cref{thm:TBNrNNIconn}, each $\rNNI$-move performed falls into
the scope of either \Cref{lem:NNI} or \Cref{lem:rootNNI}. We claim that all of
these $\rNNI$-moves performed during the proof are either already distance-1
tail-moves, or they can be simulated with tail-moves (see Definition~\ref{d:head.tail.move}).

The $\rNNI$-move described by \Cref{lem:rootNNI} is a distance-1 tail-move if
$v$ is the tail of $e$. In \Cref{sec:caseBC}, this is always the case for
applications of \Cref{lem:rootNNI}. In \Cref{sec:caseA2}, however, it is
possible that in terms of the labeling used in \Cref{lem:rootNNI}, $v$ is the
head of $e$ if $z$ is not the root. In both of these cases the $\rNNI$-move can
be simulated with two distance-1 tail-moves: first, move the tail of $zu$ onto
$vw$, and then move the tail of $zw$ onto the incoming arc of $v$ which is
different from $e$. The intermediate graph is a phylogenetic network, because
the first tail-move is an $\rNNI$-move to which \Cref{lem:NNI}\ref{lem:treeNNI}
applies. Moreover, it trivially has a tree-base, because $f$, $zu$, $zv$,
and $vw$ are all supporting the tree-base of $N$.

The $\rNNI$-move described by \Cref{lem:NNI} is a distance-1 tail-move if $v$ is
the tail of $e$ or $w$ is the tail of $f$. Otherwise, if $t(e)=v$ and $t(f)=w$,
the $\rNNI$-move can be decomposed into three tail-moves. Let $s(e)=x$ and
$s(f)=y$, so that $e=xv$ and $f=yw$. By the assumptions of
\Cref{lem:NNI}\ref{lem:treeNNI}, $y\neq u$. The arcs $e$ and $f$ are not
supporting the tree-base of $N$, hence $x$ and $y$ are not reticulation
vertices. First, move the tail of $f$ to $uv$ and let the new incoming arc of
$w$ be $y'w$. Secondly, move the tail $x$ of $e$ to the original position of
$y$. Lastly, move the tail of $y'w$ to the original position of $x$. The two
intermediate networks are trivially acyclic, because both the starting network
$N$ and the target network $N'$ are acyclic and $y'w$ is the only additional arc
in the two intermediate networks. Both of the intermediate networks possess a
tree-base, because the arcs whose tails were moved are not contained in the
chosen tree-base.

Although the three tail-moves are generally not distance-1, they can be broken
up into distance-1 tail-moves, such that the tail traverses the shortest
undirected path in the support tree. Let $N''$ be any intermediary network along
these refined steps. For any vertex $x$, the set of vertices that are
reachable through a direct path starting from $x$ is broader in $N$ (or
alternatively, in $N'$) than in $N''$. Thus $N''$ is acyclic, too.
\end{proof}

\renewcommand*{\bibfont}{\small} \printbibliography[]

\end{document}